\newif\ifconf
\conftrue 
\conffalse

\documentclass{article}
\usepackage{fullpage}
\newcommand{\myspace}[1]{}

\usepackage{theorem,latexsym,graphicx}
\usepackage{amsmath,amssymb,enumerate}
\usepackage{xspace}

\usepackage{bm}
\usepackage{ifpdf}
\usepackage{shadow,shadethm,color}
\usepackage{algorithm}
\usepackage[noend]{algorithmic}
\usepackage{subfigure}
\usepackage{verbatim}
\usepackage{paralist}

\allowdisplaybreaks

\ifconf
\newcommand{\lref}[2][]{{#1~\ref{#2}}}
\else
\definecolor{Darkblue}{rgb}{0,0,0.4}
\definecolor{Brown}{cmyk}{0,0.81,1.,0.60}
\definecolor{Purple}{cmyk}{0.45,0.86,0,0}
\newcommand{\mydriver}{hypertex}
\ifpdf
 \renewcommand{\mydriver}{pdftex}
\fi
\usepackage[breaklinks,\mydriver]{hyperref}
\hypersetup{colorlinks=true,
            citebordercolor={.6 .6 .6},linkbordercolor={.6 .6 .6},%
citecolor=Darkblue,urlcolor=black,linkcolor=Darkblue,pagecolor=black}
\newcommand{\lref}[2][]{\hyperref[#2]{#1~\ref*{#2}}}
\fi


\newtheorem{theorem}{Theorem}[section]

\newtheorem{lemma}[theorem]{Lemma}
\newshadetheorem{lemmashaded}[theorem]{Lemma}

\numberwithin{algorithm}{section}

\newenvironment{proof}{{\bf Proof:  }}{\hfill\rule{2mm}{2mm}}

\newcommand{\junk}[1]{}
\newcommand{\ignore}[1]{}

\newcommand{\R}[0]{{\ensuremath{\mathbb{R}}}}

\newcommand{\Z}[0]{{\ensuremath{\mathbb{Z}}}}

\def\ceil#1{\lceil #1 \rceil}

\newcommand{\poly}{\operatorname{poly}}

\newcommand{\sse}{\subseteq}
\newcommand{\I}{{\mathcal{I}}}

\renewcommand{\P}{{\mathcal{P}}}

\newcommand{\F}{{\mathcal{F}}}
\newcommand{\M}{{\mathcal{M}}}
\newcommand{\x}{{\mathbf{x}}}
\newcommand{\E}{{\mathbf{E}}}

\newcommand{\e}{\varepsilon}
\newcommand{\eps}{\varepsilon}
\ifconf
\renewcommand{\ts}{\textstyle}
\else
\newcommand{\ts}{\textstyle}
\fi

\newcounter{note}[section]

 \newcommand{\agnote}[1]{}
 \newcommand{\ktnote}[1]{}
 \newcommand{\dwnote}[1]{}
 \newcommand{\uwnote}[1]{}

\newcommand{\qedsymb}{\hfill{\rule{2mm}{2mm}}}

\newcommand{\initOneLiners}{%
    \setlength{\itemsep}{0pt}
    \setlength{\parsep }{0pt}
    \setlength{\topsep }{0pt}
}
\newenvironment{OneLiners}[1][\ensuremath{\bullet}]
    {\begin{list}
        {#1}
        {\initOneLiners}}
    {\end{list}}

\newcommand{\squishlist}{
 \begin{list}{$\bullet$}
  { \setlength{\itemsep}{0pt}
     \setlength{\parsep}{3pt}
     \setlength{\topsep}{3pt}
     \setlength{\partopsep}{0pt}
     \setlength{\leftmargin}{1.5em}
     \setlength{\labelwidth}{1em}
     \setlength{\labelsep}{0.5em} } }

\newcommand{\squishend}{
  \end{list}  }

\newcommand{\rank}{\textsf{rank}}
\newcommand{\spn}{\textsf{span}}

\newcommand{\MMM}{\textsf{MMM}\xspace}
\newcommand{\MSM}{\textsf{MSM}\xspace}

\begin{document}

\title{Changing Bases: Multistage Optimization for Matroids and Matchings}
\author{
Anupam Gupta\thanks{Computer Science Dept., Carnegie Mellon
    University. Research performed while the author was at Microsoft Reserach Silicon Valley. Research was partly supported by
    NSF awards CCF-0964474 and CCF-1016799, and by a grant from the
    CMU-Microsoft Center for Computational Thinking. \texttt{anupamg@cs.cmu.edu}}
\and
Kunal Talwar\thanks{Microsoft Research SVC, Mountain View, CA 94043. \texttt{kunal@microsoft.com}}
\and
Udi Wieder\thanks{Microsoft Research SVC, Mountain View, CA 94043.\texttt{uwieder@microsoft.com}}
}


\maketitle

\pagestyle{plain}

\begin{abstract}
This paper is motivated by the fact that many systems need to be
maintained continually while the underlying costs change over time. The
challenge then is to continually maintain near-optimal solutions to the
underlying optimization problems, without creating too much churn in the
solution itself. We model this as a  multistage combinatorial
optimization problem where the input is a sequence of cost functions
(one for each time step); while we can change the solution from step to
step, we incur an additional cost for every such change.

We first study the multistage matroid maintenance problem, where we need
to maintain a base of a matroid in each time step under the changing
cost functions and acquisition costs for adding new elements. The online
version of this problem generalizes onine paging, and is a
well-structured case of the metrical task systems.  E.g., given a graph,
we need to maintain a spanning tree $T_t$ at each step: we pay
$c_t(T_t)$ for the cost of the tree at time $t$, and also $| T_t
\setminus T_{t-1} |$ for the number of edges changed at this step. Our
main result is a polynomial time $O(\log m \log r)$-approximation to the
online multistage matroid maintenance problem, where $m$ is the number
of elements/edges and $r$ is the rank of the matroid. This improves on
results of Buchbinder et al.~\cite{BuchbinderCNS12} who addressed the
\emph{fractional} version of this problem under uniform acquisition
costs, and Buchbinder, Chen and Naor~\cite{BuchbinderCN14} who studied
the fractional version of a more general problem. We also give an
$O(\log m)$ approximation for the offline version of the problem.  These
bounds hold when the acquisition costs are non-uniform, in which case
both these results are the best possible unless P=NP.

We also study the perfect matching version of the problem, where we must
maintain a perfect matching at each step under changing cost functions
and costs for adding new elements. Surprisingly, the hardness
drastically increases: for any constant $\eps>0$, there is
no $O(n^{1-\eps})$-approximation to the multistage matching maintenance
problem, even in the offline case.
\end{abstract}



\vspace{-0.1in}

\section{Introduction}
\label{sec:introduction}
\myspace{-0.17in}

In a typical instance of a combinatorial optimization problem the underlying constraints model a static application frozen in one time step.
In many applications however, one needs to solve instances of the combinatorial optimization problem that changes over time. While this is naturally handled by re-solving the optimization problem in each time step separately, changing the solution one holds from one time step to the next often incurs a transition cost. Consider, for example, the problem faced by a vendor who needs to get supply of an item from $k$ different producers to meet her demand. On any given day, she could get prices from each of the producers and pick the $k$ cheapest ones to buy from. As prices change, this set of the $k$ cheapest producers may change. However, there is a fixed cost to starting and/or ending a relationship with any new producer. The goal of the vendor is to minimize the sum total of these two costs: an "acquisition cost" $a(e)$ to be incurred each time she starts a new business relationship with a producer, and a per period cost $c_t(e)$ of buying in period $t$ from the each of the $k$ producers that she picks in this period, summed over $T$ time periods. In this work we consider a generalization of this problem, where the constraint ``pick $k$  producers'' may be replaced by a more general combinatorial constraint. It is natural to ask whether simple combinatorial problems for which the one-shot problem is easy to solve, as the example above is, also admit good algorithms for the multistage version.

The first problem we study is the {\em Multistage Matroid Maintenance}
problem (\MMM), where the underlying combinatorial constraint is that of
maintaining a base of a given matroid in each period. In the example
above, the requirement the vendor buys from $k$ different producers
could be expressed as optimizing over the $k-$uniform matroid. In a more
interesting case one may want to maintain a spanning tree of a given
graph at each step, where the edge costs $c_t(e)$ change over time, and
an acquisition cost of $a(e)$ has to paid every time a new edge enters
the spanning tree. (A formal definition of the \MMM problem appears in Section~\ref{sec:formal-defs}.) While our emphasis is on the online problem, we will mention results for the offline version as well, where the whole input is given in advance.



A first observation we make is that if the matroid in question is
allowed to be different in each time period, then the problem is hard to
approximate to any non-trivial factor (see
Section~\ref{sec:time-varying}) even in the offline case. We therefore focus on the case where the same matroid is given at each time period. Thus we restrict ourselves to the case when the matroid is the same for all time steps.


To set the baseline, we first study the offline version of the problem
(in Section~\ref{sec:offline}), where all the input parameters are known
in advance. We show an LP-rounding algorithm which approximates the
total cost up to a logarithmic factor. This approximation factor is no
better than that using a simple greedy algorithm, but it will be useful to see
the rounding algorithm, since we will use its extension in the online
setting. We also show a matching hardness reduction, proving that the
problem is hard to approximate to better than a logarithmic factor; this
hardness holds even for the special case of spanning trees in
graphs.

We then turn to the online version of the problem, where in each time
period, we learn the costs $c_t(e)$ of each element that is available at
time $t$, and we need to pick a base $S_t$ of the matroid for this
period. We analyze the performance of our online algorithm in the
competitive analysis framework: i.e., we compare the cost of the online
algorithm to that of the optimum solution to the offline instance thus
generated. In Section~\ref{sec:online}, we give an efficient randomized
$O(\log |E| \log (rT))$-competitive algorithm for this problem against
any oblivious adversary (here $E$ is the universe for the matroid and
$r$ is the rank of the matroid), and show that no polynomial-time online
algorithm can do better. We also show that the requirement that the
algorithm be randomized is necessary: any deterministic algorithm must
incur an overhead of $\Omega(\min(|E|,T))$, even for the simplest of
matroids.

Our results above crucially relied on the properties of matriods, and it
is natural to ask if we can handle more general set systems, e.g.,
$p$-systems.  In Section~\ref{sec:matchings}, we consider the case where
the combinatorial object we need to find each time step is a perfect
matching in a graph. Somewhat surprisingly, the problem here is
significantly harder than the matroid case, even in the offline case. In
particular, we show that even when the number of periods is a constant,
no polynomial time algorithm can achieve an approximation ratio better
than $\Omega(|E|^{1-\epsilon})$ for any constant $\epsilon>0$.



\myspace{-0.24in}
\subsection{Techniques}
\myspace{-0.15in}

We first show that the \MMM problem, which is a packing-covering
problem, can be reduced to the analogous problem of maintaining a
spanning set of a matroid. We call the latter the {\em Multistage
  Spanning set Maintenance} (\MSM) problem. While the reduction itself
is fairly clean, it is surprisingly powerful and is what enables us to improve on previous works. The \MSM problem is a
covering problem, so it admits better approximation ratios and allows
for a much larger toolbox of techniques at our disposal. We note that
this is the only place where we need the matroid to not change over
time: our algorithms for \MSM work when the matroids change over time,
and even when considering matroid intersections. The \MSM problem is
then further reduced to the case where the holding cost of an element is
in $\{0,\infty\}$, this reduction simplifies the analysis.

In the offline case, we present two algorithms.  We first observe that a
greedy algorithm easily gives an $O(\log T)$-approximation.
We then present a simple randomized rounding algorithm for the linear
program. This is analyzed using recent results on contention resolution
schemes~\cite{CVZ11}, and gives an approximation of $O(\log rT)$, which
can be improved to $O(\log r)$ when the acquisition costs are
uniform. This LP-rounding algorithm will be an important constituent of
our algorithm for the online case.

For the online case we again use that the problem can be written as a
covering problem, even though the natural LP formulation has both
covering and packing constraints. Phrasing it as a covering problem
(with box constraints) enables us to use, as a black-box, results on
online algorithms for the fractional problem~\cite{BN-MOR}. This
formulation however has exponentially many constraints. We handle that
by showing a method of adaptively picking violated constraints such that
only a small number of constraints are ever picked. The crucial insight
here is that if $x$ is such that $2x$ is not feasible, then $x$ is at
least $\frac{1}{2}$ away in $\ell_1$ distance from any feasible
solution; in fact there is a single constraint that is violated to an
extent half. This insight allows us to make non-trivial progress (using
a natural potential function) every time we bring in a constraint, and
lets us bound the number of constraints we need to add until constraints
are satisfied by $2x$. 

\myspace{-0.2in}
\subsection{Related Work}
\label{sec:related-work}
\myspace{-0.15in}

Our work is related to several lines of research, and extends some of
them. The paging problem is a special case of \MMM
where the underlying matroid is a uniform one. Our online algorithm
generalizes the $O(\log k)$-competitive algorithm for weighted caching~\cite{BBN-stoc08-caching}, using existing
online LP solvers in a black-box fashion. Going from uniform
to general matroids loses a logarithmic
factor (after rounding), we show such a
loss is unavoidable unless we use exponential time.

The \MMM problem is also a special case of
classical Metrical Task Systems \cite{Borodin:1992:OOA:146585.146588};
see~\cite{DBLP:conf/alt/AbernethyBBS10,DBLP:conf/icalp/BansalBN10} for more recent work. The best
approximations for metrical task systems are poly-logarithmic in the size
of the metric space. In our case the metric space is specified by the
total number of bases of the matroid which is often exponential, so
these algorithms only give a trivial approximation.



In trying to unify online learning and competitive analysis, Buchbinder
et al.~\cite{BuchbinderCNS12} consider a problem on matroids very
similar to ours. The salient differences are: (a)~in their model all
acquisition costs are the same, and (b)~they work with fractional bases
instead of integral ones. They give an $O(\log n)$-competitive algorithm
to solve the fractional online LP with uniform acquisition costs (among
other unrelated results). Our online LP solving generalizes their result
to arbitrary acquisition costs. They leave open the question of getting
integer solutions online (Seffi Naor, private communication), which we
present in this work. In a more recent work, Buchbinder, Chen and
Naor~\cite{BuchbinderCN14} use a regularization approach to solving a
broader set of fractional problems, but once again can do not get
integer solutions in a setting such as
ours. 

Shachnai et
al.~\cite{Shachnai:2012:TAC:2247370.2247422} consider
``reoptimization'' problems: given a starting solution and a new instance,
they want to balance the transition cost and the cost on the new
instance. This is a two-timestep version of our problem, and the
short time horizon raises a very different set of issues (since the output
solution does not need to itself hedge against possible subsequent
futures). They consider a number of optimization/scheduling problems in
their framework.

Cohen et al.~\cite{DBLP:journals/corr/abs-1302-2137} consider several
problems in the framework of the stability-versus-fit tradeoff; e.g.,
that of finding ``stable'' solutions which given the previous solution,
like in reoptimization, is the current solution that maximizes the
quality minus the transition costs. They show maintaining stable
solutions for matroids becomes a repeated two-stage reoptimization
problem; their problem is poly-time solvable, whereas matroid problems
in our model become NP-hard. The reason is that the solution for two
time steps does not necessarily lead to a base from which it is easy to
move in subsequent time steps, as our
hardness reduction shows. They consider a multistage offline version of
their problem (again maximizing fit minus stability) which is very
similar in spirit and form to our (minimization) problem, though the
minus sign in the objective function makes it difficult to approximate
in cases which are not in poly-time.

In dynamic Steiner tree maintenance~\cite{IW91,MSVW12,GuGK13} where the
goal is to maintain an approximately optimal Steiner tree for a varying
instance (where terminals are added) while changing few edges at each
time step. In dynamic load balancing~\cite{AGZ99,EL11} one has to
maintain a good scheduling solution while moving a small number of jobs
around. The work on lazy experts in the online prediction
community~\cite{CesabianchiL06} also deals with similar concerns.


There is also work on ``leasing''
problems~\cite{DBLP:conf/focs/Meyerson05, AnthonyG07, NagarajanW08}: these are optimization problems where
elements can be obtained for an interval of any length, where the cost
is concave in the lengths; the instance changes at each timestep. The
main differences are that the solution only needs to be feasible at each
timestep (i.e., the holding costs are $\{0, \infty\}$), and that any
element can be leased for any length $\ell$ of time starting at any
timestep for a cost that depends only on $\ell$, which gives these
problems a lot of uniformity. In turn, these leasing problems are
related to ``buy-at-bulk'' problems.





\myspace{-0.2in}
\section{Maintaining Bases to Maintaining Spanning Sets}
\label{sec:prelims}
\myspace{-0.18in}

Given reals $c(e)$ for elements $e \in E$, we will use $c(S)$ for $S
\sse E$ to denote $\sum_{e \in S} c(e)$. We denote $\{1,2,\ldots, T\}$
by $[T]$.

We assume basic familiarity with matroids: see, e.g.,~\cite{Sch-book}
for a detailed treatment. Given a matroid $\M = (E, \I)$, a \emph{base}
is a maximum cardinality independent set, and a \emph{spanning set} is a
set $S$ such that $\rank(S) = \rank(E)$; equivalently, this set contains
a base within it. The \emph{span} of a set $S \sse E$ is $\spn(S) = \{ e
\in E \mid \rank(S + e) = \rank(S) \}$. The \emph{matroid polytope}
$\P_I(\M)$ is defined as $\{ x \in \R^{|E|}_{\geq 0} \mid x(S) \leq
\rank(S) \,\,\forall S \sse E \}$. The \emph{base polytope} $\P_B(\M) =
\P_I(\M) \cap \{ x \mid x(E) = \rank(E) \}$. We will sometimes use $m$ to denote $|E|$ and $r$ to denote the rank of the matroid.
\myspace{-0.2in}

\subsubsection*{Formal Definition of Problems}
\label{sec:formal-defs}
\myspace{-0.17in}

An instance of the {\em Multistage Matroid Maintenance} (\MMM) problem consists of a matroid $\M = (E, \I)$, an
\emph{acquisition cost} $a(e)\geq 0$ for each $e \in E$, and for every
timestep $t \in [T]$ and element $e \in E$, a \emph{holding cost} cost
$c_t(e)$. The goal is to find bases $\{ B_t \in \I \}_{t \in [T]}$  to
minimize
\begin{gather}
  \ts \sum_t \big( c_t(B_t) + a(B_t \setminus B_{t-1}) \big),
\end{gather}
where we
define $B_0 := \emptyset$. A related problem is the {\em Multistage Spanning set Maintenance}(\MSM) problem, where
we want to maintain a spanning set $S_t \sse E$ at each time, and cost
of the solution $\{S_t\}_{t \in [T]}$ (once again with $S_0 := \emptyset$) is
\begin{gather}
  \ts \sum_t \big( c_t(S_t) + a(S_t \setminus S_{t-1}) \big).
\end{gather}

\myspace{-0.2in}
\subsubsection*{Maintaining Bases versus Maintaining Spanning Sets}
\label{sec:cover-vs-pack}
\myspace{-0.17in}
The following lemma shows the equivalence of maintaining bases and spanning sets.  This enables us to significantly simplify the problem and avoid the difficulties faced by previous works on this problem.
\begin{lemma}
  \label{lem:pack-cover}
  For matroids, the optimal solutions to \MMM and \MSM have the same
  costs.
\end{lemma}

\begin{proof}
  Clearly, any solution to \MMM is also a solution to \MSM, since a base
  is also a spanning set. Conversely, consider a solution $\{ S_t \}$ to
  \MSM. Set $B_1$ to any base in $S_1$. Given $B_{t-1} \sse S_{t-1}$,
  start with $B_{t-1} \cap S_t$, and extend it to any base $B_t$ of
  $S_t$. This is the only step where we use the matroid
  properties---indeed, since the matroid is the same at each time, the
  set $B_{t-1} \cap S_t$ remains independent at time $t$, and by the
  matroid property this independent set can be extended to a base.
  Observe that this process just requires us to know the base $B_{t-1}$
  and the set $S_t$, and hence can be performed in an online fashion.

  We claim that the cost of $\{ B_t \}$ is no more than that of $\{ S_t
  \}$. Indeed, $c_t(B_t) \leq c_t(S_t)$, because $B_t \sse
  S_t$. Moreover, let $D := B_t \setminus B_{t-1}$, we pay $\sum_{e \in
    D} a_e$ for these elements we just added. To charge this, consider
  any such element $e \in D$, let $t^\star \leq t$ be the time it was
  most recently added to the cover---i.e., $e \in S_{t'}$ for all $t'
  \in [t^\star, t]$, but $e \not\in S_{t^\star - 1}$. The \MSM solution
  paid for including $e$ at time $t^\star$, and we charge our
  acquisition of $e$ into $B_t$ to this pair $(e, t^\star)$. It suffices
  to now observe that we will not charge to this pair again, since the
  procedure to create $\{ B_t\}$ ensures we do not drop $e$ from the
  base until it is dropped from $S_t$ itself---the next time we pay an
  addition cost for element $e$, it would have been dropped and added in
  $\{ S_t\}$ as well.
\end{proof}

Hence it suffices to give a good solution to the \MSM problem. We observe that the
proof above uses the matroid property crucially and would not hold,
e.g., for matchings. It also requires that the \emph{same} matroid be
given at all time steps. Also, as noted above, the reduction is online: the instance is the same, and given an \MSM solution it can be transformed online to a solution to \MMM.


\myspace{-0.1in}
\subsubsection*{Elements and Intervals}
\label{sec:intervals}
\myspace{-0.1in}

We will find it convenient to think of an instance of \MSM as being a
matroid $\M$, where each element only has an acquisition cost $a(e) \geq
0$, and it has a lifetime $I_e = [l_e,r_e]$. There are no holding costs,
but the element $e$ can be used in spanning sets only for timesteps $t
\in I_e$. Or one can equivalently think of holding costs being zero for
$t \in I_e$ and $\infty$ otherwise.

\emph{An Offline Exact Reduction.}  The translation is the natural one:
given instance $(E, \I)$ of \MSM, create elements $e_{lr}$ for each $e
\in E$ and $1 \leq l \leq r \leq T$, with acquisition cost $a(e_{lr}) :=
a(e) + \sum_{t = l}^{r} c_t(e)$, and interval $I_{e_{lr}} :=
[l,r]$. (The matroid is extended in the natural way, where all the
elements $e_{lr}$ associated with $e$ are parallel to each other.)  The
equivalence of the original definition of \MSM and this interval view is
easy to verify.

\emph{An Online Approximate Reduction.}  Observe that the above reduction created at most $\binom{T}{2}$ copies of each element, and required knowledge of all the costs. If we are willing to lose a constant factor in the approximation,
we can perform a reduction to the interval model in an \emph{online}
fashion as follows. For element $e \in E$, define $t_0 = 0$, and create
many parallel copies $\{e_i\}_{i \in \Z_+}$ of this element (modifying
the matroid appropriately). Now the $i^{th}$ interval for $e$ is
$I_{e_i} := [t_{i-1}+1, t_i]$, where $t_i$ is set to $t_{i-1} + 1$ in
case $c_{t_{i-1}+1}(e) \geq a(e)$, else it is set to the \emph{largest}
time such that the total holding costs $\sum_{t = t_{i-1}+1}^{t_i}
c_t(e)$ for this interval $[t_{i-1}+1, t_i]$ is at most $a(e)$. This
interval $I_{e_i}$ is associated with element $e_i$, which is only
available for this interval, at cost  $a(e_i) = a(e) +
c_{t_{i-1}+1}(e)$.

A few salient points about this reduction: the intervals for an original
element $e$ now partition the entire time horizon $[T]$. The number of
elements in the modified matroid whose intervals contain any time $t$ is
now only $|E| = n$, the same as the original matroid; each element of
the modified matroid is only available for a single interval. Moreover, the
reduction can be done online: given the past history and the holding
cost for the current time step $t$, we can ascertain whether $t$ is the
beginning of a new interval (in which case the previous interval ended
at $t-1$) and if so, we know the cost of acquiring a copy of $e$ for the
new interval is $a(e) + c_t(e)$. It is easy to check that the optimal
cost in this interval model is within a constant factor of the optimal
cost in the original acquisition/holding costs model.

\myspace{-0.2in}
\section{Offline Algorithms}
\label{sec:offline}
\myspace{-0.17in}


Given the reductions of the previous section, we can focus on the \MSM
problem. Being a covering problem, \MSM is conceptually easier to solve:
e.g., we could use algorithms for submodular set cover~\cite{Wol82} with
the submodular function being the sum of ranks at each of the timesteps,
to get an $O(\log T)$ approximation.

In Section~\ref{sec:greedy}, we give a dual-fitting proof of the
performance of the greedy algorithm. 
Here we give an LP-rounding algorithm which gives an $O(\log
rT)$ approximation; this can be improved to $O(\log r)$ in the common
case where all acquisition costs are unit. (While the approximation
guarantee is no better than that from submodular set cover, this
LP-rounding algorithm will prove useful in the online case in
Section~\ref{sec:online}).  Finally, the hardness results of
Section~\ref{sec:hardness-offline} show that we cannot hope to do much
better than these logarithmic approximations.

\myspace{-0.2in}

\subsection{The LP Rounding Algorithm}
\label{sec:lp-round}
\myspace{-0.15in}


We now consider an LP-rounding algorithm for the \MMM problem; this will
generalize to the online setting, whereas it is unclear how to extend
the  greedy algorithm to that case. For the LP rounding, we use the
standard definition of the \MMM problem to write the following LP
relaxation.
\begin{align}
  \min \sum_{t,e} a(e) \cdot y_t(e) &+ \sum_{t,e} c_t(e) \cdot z_{t}(e)
  \tag{LP2} \label{eq:lp2} \\
  \text{s.t.~~~} z_{t} &\in \P_B(\M) \qquad\qquad\qquad \forall t \notag
  \\
  y_t(e) &\geq z_{t}(e) - z_{t-1}(e) \qquad \forall t, e \notag\\
  y_t(e), z_{t}(e) &\geq 0 \notag
\end{align}
It remains to round the solution to get a feasible solution to \MSM
(i.e., a spanning set $S_t$ for each time) with expected cost at most
$O(\log n)$ times the LP value, since we can use
Lemma~\ref{lem:pack-cover} to convert this to a solution for \MMM at no
extra cost.
The following lemma is well-known
 (see, e.g.~\cite{CalinescuCPV07}). We give a proof for completeness. \ktnote{Is this the best reference?}
\begin{lemma}
  \label{lem:alon}
  For a fractional base $z \in \P_B(M)$, let $R(z)$ be the set obtained
  by picking each element $e \in E$ independently with probability
  $z_e$. Then $E[\rank(R(z))] \geq r(1-1/e)$.
\end{lemma}
\begin{proof}
  We use the results of Chekuri et al.~\cite{CVZ11} (extending those of
  Chawla et al.~\cite{CHMS10}) on so-called contention resolution schemes.
  Formally, for a matroid $\M$, they give a randomized procedure $\pi_z$
  that takes the random set $R(z)$ and outputs an independent set
  $\pi_z(R(z))$ in $\M$, such that $\pi_z(R(z)) \sse R(z)$, and for each
  element $e$ in the support of $z$, $\Pr[ e \in \pi_z(R(z)) \mid e \in
  R(z) ] \geq (1-1/e)$. (They call this a $(1, 1-1/e)$-balanced CR
  scheme.) Now, we get
  \begin{align*}
    \E[ \rank(R(z)) ] &\geq \E[ \rank(\pi_z(R(z))) ] = \sum_{e \in \text{supp}(z)} \Pr[ e \in
    \pi_z(R(z)) ] \\
    & = \sum_{e \in \text{supp}(z)} \Pr[ e \in \pi_z(R(z)) \mid e \in R(z) ] \cdot \Pr[ e \in R(z) ]
    \\
    & \geq \sum_{e \in \text{supp}(z)} (1-1/e) \cdot z_e = r(1-1/e).
  \end{align*}
  The first inequality used the fact that $\pi_z(R(z))$ is a subset of
  $R(z)$, the following equality used that $\pi_z(R(z))$ is independent
  with probability~1, the second inequality used the property of the CR
  scheme, and the final equality used the fact that $z$ was a fractional
  base.
\end{proof}

\begin{theorem}
  \label{thm:lp-round}
  Any fractional solution can be randomly rounded to get solution to
  \MSM with cost $O(\log rT)$ times the fractional value, where $r$ is
  the rank of the matroid and $T$ the number of timesteps.
\end{theorem}

\begin{proof}
  Set $L = 32 \log (rT)$. For each element $e \in E$, choose a random
  threshold $\tau_e$ independently and uniformly from the interval $[0,
  1/L]$. For each $t \in T$, define the set $\widehat{S}_t := \{ e \in E
  \mid z_t(e) \geq \tau_e\}$; if $\widehat{S}_t$ does not have full
  rank, augment its rank using the cheapest elements according to
  $(c_t(e) + a(e))$ to obtain a full rank set $S_t$. Since $\Pr[ e \in
  \widehat{S}_t ] = \min\{ L\cdot z_t(e), 1\}$, the cost
  $c_t(\widehat{S}_t) \leq L \times (c_t \cdot z_t)$. Moreover, $e \in
  \widehat{S}_t \setminus \widehat{S}_{t-1}$ exactly when $\tau_e$
  satisfies $z_{t-1}(e) < \tau_e \leq z_t(e)$, which happens with
  probability at most
  \[ \frac{\max( z_t(e) - z_{t-1}(e), 0)}{1/L} \le L\cdot y_t(e).\]
  Hence the expected acquisition cost for the elements newly added to
  $\widehat{S}_t$ is at most $L \times \sum_e (a(e) \cdot
  y_t(e))$. Finally, we have to account for any elements added to extend
  $\widehat{S}_t$ to a full-rank set $S_t$.

  \begin{lemma}
    \label{lem:rand-round}
    For any fixed $t \in [T]$, the set $\widehat{S}_t$ contains a basis of
    $\M$ with probability at least $1 - 1/(rT)^8$.
  \end{lemma}

  \begin{proof}
    The set $\widehat{S}_t$ is obtained by threshold rounding of the
    fractional base $z_t \in \P_B(\M)$ as above. Instead, consider
    taking $L$ different samples $T^{(1)}, T^{(2)}, \ldots, T^{(L)}$,
    where each sample is obtained by including each element $e \in E$
    independently with probability $z_t(e)$; let $T := \cup_{i = 1}^L
    T^{(i)}$. It is easy to check that $\Pr[ \rank(T) = r ] \leq \Pr[
    \rank(\widehat{S}_t) = r]$, so it suffices to give a lower bound on
    the former expression. For this, we use Lemma~\ref{lem:alon}: the
    sample $T^{(1)}$ has expected rank $r(1-1/e)$, and using reverse
    Markov, it has rank at least $r/2$ with probability at least $1-2/e
    \geq 1/4$. Now focusing on the matroid $\M'$ obtained by contracting
    elements in $\spn(T^{(1)})$ (which, say, has rank $r'$), the same
    argument says the set $T^{(2)}$ has rank $r'/2$ with probability at
    least $1/4$, etc. Proceeding in this way, the probability that the
    rank of $T$ is less than $r$ is at most the probability that we see
    fewer than $\log_2 r$ heads in $L = 32\log rT$ flips of a
    coin of bias $1/4$. By a Chernoff bound, this is at
    most $\exp\{-(7/8)^2\cdot (L/4)/3\} = 1/(rT)^8$.
  \end{proof}

  Now if the set $\widehat{S}_t$ does not have full rank, the elements we
  add have cost at most that of the min-cost base under the cost
  function $(a_e + c_t(e))$, which is at most the optimum value
  for~(\ref{eq:lp2}). (We use the fact that the LP is exact for a single
  matroid, and the global LP has cost at least the single timestep
  cost.) This happens with probability at most $1/(rT)^8$, and hence the
  total expected cost of augmenting $\widehat{S}_t$ over all $T$
  timesteps is at most $O(1)$ times the LP value. This proves the main
  theorem.
\end{proof}

Again, this algorithm for \MSM works with different matroids at each
timestep, and also for intersections of matroids. To see this observe that the only requirements from the algorithm are that there is a separation oracle for the polytope and that the contention resolution scheme works. In the case of $k-$matroid intersection, if we pay an extra $O(\log k)$ penalty in the approximation ratio we have that the probability a rounded solution does not contain a base is $< 1/k$ so we can take a union bound over the multiple matroids.

\ifconf
\else
\paragraph{An Improvement: Avoiding the Dependence on $T$.}
When the ratio of the maximum to the minimum acquisition cost is small,
we can improve the approximation factor above. More specifically, we
show that essentially the same randomized rounding algorithm (with a
different choice of $L$) gives an approximation ratio of
$\log\frac{ra_{max}}{a_{min}}$. We defer the argument to
Section~\ref{sec:just-logr}, as it needs some additional definitions and
results that we present in the online section.
\fi

\ifconf
\else
\subsection{Hardness for Offline \MSM}
\label{sec:hardness-offline}

\begin{theorem}
  \label{thm:matr-hard}
  The \MSM and \MMM problems are NP-hard to approximate better than
  $\Omega(\min\{ \log r, \log T \})$ even for graphical matroids.
\end{theorem}

\begin{proof}
  We give a reduction from Set Cover to the \MSM problem for graphical
  matroids. Given an instance $(U, \F)$ of set cover, with $m = |\F|$
  sets and $n = |U|$ elements, we construct a graph as follows. There is
  a special vertex $r$, and $m$ set vertices (with vertices $s_i$ for
  each set $S_i \in \F$). There are $m$ edges $e_i := (r, s_i)$ which
  all have inclusion weight $a(e_i) = 1$ and per-time cost $c_t(e) = 0$
  for all $t$. All other edges will be zero cost short-term edges as
  given below.  In particular, there are $T = n$ timesteps. In timestep
  $j \in [n]$, define subset $F_j := \{ s_i \mid S_i \ni u_j\}$ to be
  vertices corresponding to sets containing element $u_j$. We have a set
  of edges $(e_i, e_{i'})$ for all $i,i' \in F_j$, and all edges $(x,y)$
  for $x, y \in \{r\} \cup \overline{F_j}$. All these edges have zero
  inclusion weight $a(e)$, and are only alive at time $j$. (Note this
  creates a graph with parallel edges, but this can be easily fixed by
  subdividing edges.)

  In any solution to this problem, to connect the vertices in $F_j$ to
  $r$, we must buy some edge $(r, s_i)$ for some $s_i \in F_j$. This is
  true for all $j$, hence the root-set edges we buy correspond to a set
  cover. Moreover, one can easily check that if we acquire edges $(r,s_i)$ such that the sets $\{S_i: (r,s_i)\mbox{ acquired}\}$ form a set cover, then we can always augment using zero cost edges to get a spanning tree. Since the only edges we pay for are the $(r,s_i)$ edges, we should buy
  edges corresponding to a min-cardinality set cover, which is hard to
  approximate better than $\Omega(\log n)$. Finally, that the number of
  time periods is $T = n$, and the rank of the matroid is $m = \poly(n)$
  for these hard instances. This gives us the claimed hardness.
\end{proof}
\fi



\newcommand{\xt}{\widetilde{\x}}
\newcommand{\ALP}{\mathcal{A}_{onLP}}

\myspace{-0.2in}
\section{Online \MSM}
\label{sec:online}
\myspace{-0.17in}

We now turn to solving \MMM in the online setting. In this setting, the
acquisition costs $a(e)$ are known up-front, but the holding costs
$c_t(e)$ for day $t$ are not known before day $t$. Since the equivalence
given in Lemma~\ref{lem:pack-cover} between \MMM and \MSM holds even in
the online setting, we can just work on the \MSM problem. We show that
the online \MSM problem admits an $O(\log |E| \log rT)$-competitive
(oblivious) randomized algorithm. To do this, we show that one can find
an $O(\log |E|)$-competitive fractional solution to the linear programming
relaxation in Section~\ref{sec:offline}, and then we round this LP
relaxation online, losing another logarithmic factor.
\myspace{-0.2in}

\subsection{Solving the LP Relaxations Online}
\label{sec:lp-online}
\myspace{-0.15in}

Again, we work in the interval model outlined in
Section~\ref{sec:intervals}. Recall that in this model, for each element
$e$ there is a unique interval $I_e \sse [T]$ during which it is alive.
The element $e$ has an acquisition cost $a(e)$, no holding costs. Once
an element has been acquired (which can be done at any time during its
interval), it can be used at all times in that interval, but not after
that. In the online setting, at each time step $t$ we are told which
intervals have ended (and which have not); also, which new elements $e$
are available starting at time $t$, along with their acquisition costs
$a(e)$. Of course, we do not know when its interval $I_e$ will end; this
information is known only once the interval ends.

We will work with the same LP as in Section~\ref{sec:lp-round}, albeit
now we have to solve it online. The variable $x_e$ is the indicator for
whether we acquire element~$e$.
\begin{align}
  P := \min \ts \sum_e a(e) \cdot x_e & \tag{LP3} \label{eq:3} \\
  \text{s.t.~~~} z_{et} &\in \P_B(\M) \qquad\qquad\qquad \forall t \notag\\
    z_{et} &\leq x_e \qquad\qquad\qquad\forall e, t \in I_e \notag\\
    x_e, z_{et} &\in [0,1] \notag
\end{align}
Note that this is not a packing or covering LP, which makes it more
annoying to solve online. Hence we consider a slight reformulation.  Let
$\P_{ss}(\M)$ denote the \emph{spanning set polytope} defined as the
convex hull of the full-rank (a.k.a.\ spanning) sets $\{ \chi_S \mid S
\sse E, \rank(S) = r\}$. Since each spanning set contains a base, we can
write the constraints of~(\ref{eq:3}) as:
\begin{align}
\x_{E_t} &\in \P_{ss}(\M) \qquad\qquad \forall t, {\mbox{ where }}  E_t =
   \{e: t \in I_e\}. \label{eq:4}
\end{align}
Here we define $\x_S$ to be the vector derived from $\x$ by zeroing out
the $x_e$ values for $e \not\in S$. It is known that the polytope
$\P_{ss}(\M)$ can be written as a (rather large) set of covering
constraints. Indeed, $\x \in \P_{ss}(\M) \iff (\mathbf{1}-\x) \in
\P_I(\M^*)$, where $\M^*$ is the dual matroid for $\M$. Since the rank
function of $M^*$ is given by $r^*(S) = r(E\setminus S) + |S| - r(E)$,
it follows that~(\ref{eq:4}) can be written as
\begin{align}
\ts \sum_{e\in S} x_e &\geq r(E) - r(E\setminus S)\qquad\qquad\qquad \forall
t, \forall S \subseteq E_t \tag{LP4} \label{eq:coveringconstraints}\\
x_e &\geq 0\qquad\qquad \forall e \in E\notag\\
x_e &\leq 1\qquad\qquad \forall e \in E\notag.
\end{align}
Thus we get a covering LP with ``box'' constraints over $E$. The
constraints can be presented one at a time: in timestep $t$, we present
all the covering constraints corresponding to $E_t$. We remark that the newer machinery of~\cite{BuchbinderCN14} may be applicable to \ref{eq:coveringconstraints}. We next show that a simpler approach suffices\footnote{Additionally, Lemma~\ref{lem:farfromfeasible} will be useful in improving the rounding algorithm.}. The general results
of Buchbinder and Naor~\cite{BN-MOR} (and its extension to row-sparse
covering problems by~\cite{GN12-online}) imply a deterministic algorithm
for fractionally solving this linear program online, with a competitive
ratio of $O(\log |E|) = O(\log m)$. However, this is not yet a
polynomial-time algorithm, the number of constraints for each timestep
being exponential.  We next give an adaptive algorithm to generate a
small yet sufficient set of constraints.
\myspace{-0.2in}
\ifconf
\else
\paragraph{Solving the LP Online in Polynomial Time.}
\fi
Given a vector $\x \in [0,1]^E$, define $\xt$ as follows:
\begin{align}
\tilde{x}_e = \min(2\,x_e,1) \qquad \forall e \in E. \label{eq:5}
\end{align}
Clearly, $\xt \leq 2\x$ and $\xt \in [0,1]^E$. We next describe the
algorithm for generating covering constraints in timestep $t$. Recall
that~\cite{BN-MOR} give us an online algorithm $\ALP$ for solving a
fractional covering LP with box constraints; we use this as a
black-box. (This LP solver only raises variables, a fact we will use.)
In timestep $t$, we adaptively select a small subset of the covering
constraints from (\ref{eq:coveringconstraints}), and present it to
$\ALP$. Moreover, given a fractional solution returned by $\ALP$, we
will need to massage it at the end of timestep $t$ to get a solution
satisfying all the constraints from (\ref{eq:coveringconstraints})
corresponding to $t$.

Let $\x$ be the fractional solution to~(\ref{eq:coveringconstraints}) at
the end of timestep $t-1$. Now given information about timestep $t$, in
particular the elements in $E_t$ and their acquisition costs, we do the
following. Given $\x$, we construct $\xt$ and check if $\xt_{E_t} \in
\P_{ss}(M)$, as one can separate for $\P_{ss}(\M)$. If $\xt_{E_t} \in
\P_{ss}(M)$, then $\xt$ is feasible and we do not need to present any
new constraints to $\ALP$, and we return $\xt$. If not, our separation
oracle presents an $S$ such that the constraint $\sum_{e\in S}
\widetilde{x}_e \geq r(E) - r(E\setminus S)$ is violated. We present the
constraint corresponding to $S$ to $\ALP$ to get an updated $\x$, and
repeat until $\xt$ is feasible for time $t$. (Since $\ALP$ only raises
variables and we have a covering LP, the solution remains feasible for
past timesteps.) We next argue that we do not need to repeat this loop
more than $2n$ times.

\begin{lemma}
  \label{lem:farfromfeasible}
  If for some $\x$ and the corresponding $\xt$, the constraint
  $\sum_{e\in S} \widetilde{x}_e \geq r(E) - r(E\setminus S)$ is
  violated. Then
\begin{gather*}
  \ts \sum_{e\in S} x_e \leq r(E) - r(E\setminus S) - \frac{1}{2}
\end{gather*}
\end{lemma}
\begin{proof}
  Let $S_1 = \{e \in S : \widetilde{x}_e = 1\}$ and let $S_2 =
  S\setminus S_1$. Let $\gamma$ denote $\sum_{e\in S_2}
  \widetilde{x}_e$. Thus
  \begin{align*}
    \ts |S_1| = \sum_{e\in S} \widetilde{x}_e -  \sum_{e\in S_2} \widetilde{x}_e
    < r(E) - r(E\setminus S) - \gamma
  \end{align*}
  Since both $|S_1|$ and $r(E) - r(E\setminus S)$ are integers, it
  follows that $|S_1| \leq r(E) - r(E\setminus S) - \ceil{\gamma}$.  On
  the other hand, for every $e \in S_2, x_e = \frac{1}{2}\cdot
  \widetilde{x}_e$, and thus $\sum_{e \in S_2} x_e =
  \frac{\gamma}{2}$. Consequently
  \begin{align*}
    \ts \sum_{e \in S} x_e  &= \ts \sum_{e \in S_1} x_e  + \sum_{e \in S_2} x_e
    = |S_1| + \frac{\gamma}{2}\\
    &\le \ts r(E) - r(E\setminus S) -  \ceil{\gamma} + \frac{\gamma}{2}.
  \end{align*}
  Finally, for any $\gamma >0$, $\ceil{\gamma}-\frac{\gamma}{2} \geq
  \frac{1}{2}$, so the claim follows.
\end{proof}

The algorithm $\ALP$ updates $\x$ to satisfy the constraint given to it,
and Lemma~\ref{lem:farfromfeasible} implies that each constraint we give
to it must increase $\sum_{e \in E_t} x_e$ by at least
$\frac{1}{2}$. The translation to the interval model ensures that the
number of elements whose intervals contain $t$ is at most $|E_t| \leq
|E| = m$, and hence the total number of constraints presented at any
time $t$ is at most $2m$. We summarize the discussion of this section in
the following theorem.
\begin{theorem}
  \label{thm:lp-solve}
  There is a polynomial-time online algorithm to compute an $O(\log
  |E|)$-approximate solution to~(\ref{eq:3}).
\end{theorem}

We observe that the solution to this linear program can be trivially transformed to one for the LP in Section~\ref{sec:lp-round}. Finally, the randomized rounding algorithm of Section~\ref{sec:lp-round}
can be implemented online by selecting a threshold $t_e \in [0,1/L]$ the
beginning of the algorithm, where $L = \Theta(\log rT)$ and selecting
element $e$ whenever $\widetilde{x}_e$ exceeds $t_e$: here we use the fact that the online algorithm only ever raises $x_e$ values, and this rounding algorithm is monotone. Rerandomizing in
case of failure gives us an expected cost of $O(\log rT)$ times the LP
solution, and hence we get an $O(\log m \log rT)$-competitive algorithm.

\ifconf
\else
\myspace{-0.1in}

\subsection{An $O(\log r\frac{a_{max}}{a_{min}})$-Approximate Rounding}
\label{sec:just-logr}
\myspace{-0.1in}

The dependence on the time horizon $T$ is unsatisfactory in some
settings, but we can do better using Lemma~\ref{lem:farfromfeasible}.
Recall that the $\log (rT)$-factor loss in the rounding follows from the
naive union bound over the $T$ time steps. We now argue that when $
\frac{a_{max}}{a_{min}}$ is small, we can afford for the rounding to
fail occasionally, and charge it to the acquisition cost incurred by the
linear program.

Let us divide the period $[1\ldots T]$ into disjoint ``epochs'', where
an epoch (except for the last) is an interval $[p,q)$ for $p \leq q$
such that the total fractional acquisition cost $\sum_{t = p}^{q-1}
\sum_e a(e) \cdot y_t(e) \geq r\cdot a_{max} > \sum_{t = p}^{q-2} \sum_e
a(e) \cdot y_t(e)$. Thus an epoch is a minimal interval where the linear
program spends acquisition cost $\in [r\cdot a_{max}, 2r \cdot a_{\max}]$,
so that we can afford to build a brand new tree once in each epoch and
can charge it to the LP's fractional acquisition cost in the
epoch. Naively applying Theorem~\ref{thm:lp-round} to each epoch
independently gives us a guarantee of $O(\log rT')$, where $T'$ is the
maximum length of an epoch.

However, an epoch can be fairly long if the LP solution changes very
slowly. We break up each epoch into phases, where each phase is a
maximal subsequence such that the LP incurs acquisition cost at most
$\frac{a_{min}}{4}$; clearly the epoch can be divided into at most
$R:= \frac{8ra_{max}}{a_{min}}$ disjoint phases. For a phase $[t_1,t_2]$,
let $Z_{[t_1,t_2]}$ denote the solution defined as $Z_{[t_1,t_2]}(e) =
\min_{t\in [t_1,t_2]} z_t(e)$. The definition of the phase implies that
for any $t\in [t_1,t_2]$, the $L_1$ difference $\|
Z_{[t_1,t_2]} - z_t\|_1 \leq \frac{1}{4}$. Now Lemma~\ref{lem:farfromfeasible}
implies that $\widetilde{Z}_{[t_1,t_2]}$ is in $\P_{ss}(M)$, where
$\widetilde{Z}$ is defined as in~(\ref{eq:5}).

Suppose that in the randomized rounding algorithm, we pick the threshold
$t_e \in [0,1/L']$ for $L' = 64 \log R$. Let $\mathcal{G}_{[t_1,t_2]}$
be the event that the rounding algorithm applied to $Z_{[t_1,t_2]}$
gives a spanning set. Since $\widetilde{Z}_{[t_1,t_2]} \leq
2Z_{[t_1,t_2]}$ is in $\P_{D}(M)$ for a phase $[t_1,t_2]$,
Lemma~\ref{lem:rand-round} implies that the event
$\mathcal{G}_{[t_1,t_2]}$ occurs with probability $1-1/R^8$. Moreover,
if $\mathcal{G}_{[t_1,t_2]}$ occurs, it is easy to see that the
randomized rounding solution is feasible for all $t\in[t_1,t_2]$. Since
there are $R$ phases within an epoch, the expected number of times that
the randomized rounding fails any time during an epoch is $R \cdot 1/R^8
= R^{-7}$.

Suppose that we rerandomize all thresholds whenever the randomized
rounding fails. Each rerandomization will cost us at most $r a_{max}$ in
expected acquisition cost. Since the expected number of times we do this
is less than once per epoch, we can charge this additional cost to the
$r a_{max}$ acquisition cost incurred by the LP during the epoch. Thus
we get an $O(\log R) = O(\log \frac{r\,a_{max}}{a_{min}})$-approximation.
This argument also works for the online case; hence for the common case
where all the acquisition costs are the same, the loss due to randomized
rounding is $O(\log r)$.

\fi

%
%
%
%
%

\ifconf
\else

\subsection{Hardness of the online \MMM and online \MSM}\label{sec:hardness-online}
In the online set cover problem, one is given an instance $(U,\F)$ of set cover, and in  time step $t$, the algorithm is presented an element $u_t \in U$, and is required to pick a set covering it. The competitive ratio of an algorithm on a sequence $\{u_t\}_{t\in [n\rq{}]}$ is the ratio of the number of sets picked by the algorithm to the optimum set-cover of the instance $(\{u_t: t \in [n\rq{}]\},\F)$. Korman~\cite[Theorem~2.3.4]{Korman05} shows the following hardness for online set cover:
\begin{theorem}[\cite{Korman05}]
There exists a constant $d>0$ such that if there is a (possibly randomized) polynomial time algorithm for online set cover with competitive ratio $d\log m \log n$,  then $NP \subseteq BPP$.
\end{theorem}

Recall that in the reduction in the proof of Theorem~\ref{thm:matr-hard}, the set of long term edges depends only on $\F$. The short term edges alone depend on the elements to be covered. It can then we verified that the same approach gives a reduction from online set cover to online \MSM. It follows that the online \MSM problem does not admit an algorithm with competitive ratio better than $d\log m \log T$ unless $NP \subseteq BPP$. In fact this hardness holds even when the end time of each edge is known as soon as it appears, and the only non-zero costs are $a(e) \in \{0,1\}$.

\fi



\myspace{-0.2in}
\section{Perfect Matching Maintenance}
\label{sec:matchings}
\myspace{-0.17in}

\newcommand{\PMM}{\textsf{PMM}\xspace}

We next consider the {\em Perfect Matching Maintenance} (\PMM) problem where $E$ is the set of edges of a graph $G = (V,E)$,
and the at each step, we need to maintain a perfect matchings in $G$.

\ifconf
\else
The natural LP
relaxation is:
\begin{align*}
  \min \sum_t c_t \cdot \x_t &+ \sum_{t,e} y_t(e) \\
  \text{s.t.~~~} \x_t &\in PM(G) \qquad\qquad\qquad \forall t \\
  y_t(e) &\geq x_t(e) - x_{t+1}(e) \qquad \forall t, e\\
  y_t(e) &\geq x_{t+1}(e) - x_t(e) \qquad \forall t, e\\
  x_t(e), y_t(e) &\geq 0
\end{align*}
The polytope $PM(G)$ is now the perfect matching polytope for $G$.

\begin{figure}[t]
\centering
\includegraphics[scale=0.5]{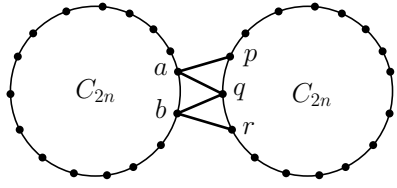}
\caption{Integrality gap example}
\label{fig:stable1}
\end{figure}

\begin{lemma}
  \label{lem:int-gap}
  There is an $\Omega(n)$ integrality gap for the \PMM problem.
\end{lemma}

\begin{proof}
  Consider the instance in the figure, and the following LP solution for
  4 time steps. In $x_1$, the edges of each of the two cycles has $x_e =
  1/2$, and the cross-cycle edges have $x_e = 0$. In $x_2$, we have
  $x_2(ab) = x_2(pq) = 0$ and $x_2(ap) = x_2(bq) = 1/2$, and otherwise
  it is the same as $x_1$. $x_3$ and $x_5$ are the same as $x_1$. In
  $x_4$, we have $x_4(ab) = x_4(qr) = 0$ and $x_4(aq) = x_4(br) = 1/2$,
  and otherwise it is the same as $x_1$.
  For each time $t$, the edges in the support of the solution
  $x_t$ have zero cost, and other edges have infinite cost. The only
  cost incurred by the LP is the movement cost, which is $O(1)$.

  Consider the perfect matching found at time $t = 1$, which must
  consist of matchings on both the cycles. (Moreover, the matching in
  time 3 must be the same, else we would change $\Omega(n)$ edges.)
  Suppose this matching uses exactly one edge from $ab$ and $pq$. Then
  when we drop the edges $ab, pq$ and add in $ap, bq$, we get a cycle on
  $4n$ vertices, but to get a perfect matching on this in time $2$ we
  need to change $\Omega(n)$ edges. Else the matching uses exactly one
  edge from $ab$ and $qr$, in which case going from time $3$ to time $4$
  requires $\Omega(n)$ changes.
\end{proof}

\fi


\ifconf
\else
\subsection{Hardness of PM-Maintenance}
\label{sec:match-hard}

In this section we prove the following hardness result:
\begin{theorem}
  For any $\e > 0$ it is NP-hard to distinguish \PMM instances with cost
  $N^\e$ from those with cost $N^{1-\e}$, where $N$ is the number of
  vertices in the graph. This holds even when the holding costs are in
  $\{0,\infty\}$, acquisition costs are $1$ for all edges, and the number
  of time steps is a constant.
\end{theorem}

\begin{proof}
  The proof is via reduction from $3$-coloring. We assume we are given
  an instance of $3$-coloring $G = (V,E)$ where the maximum degree of
  $G$ is constant. It is known that the $3$-coloring problem is still
  hard for graphs with bounded degree~\cite[Theorem~2]{GuruK04}.


  \begin{figure}[h]
    \centering
    \includegraphics[scale=0.5]{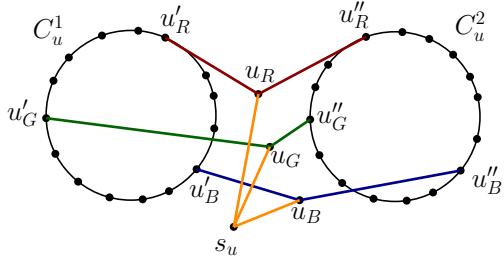}
    \caption{Per-vertex gadget}
    \label{fig:gadget}
  \end{figure}

  We construct the following gadget $X_u$ for each vertex $u \in V$. (A
  figure is given in Figure~\ref{fig:gadget}.)
  \begin{OneLiners}
  \item There are two cycles of length $3\ell$, where $\ell$ is odd. The
    first cycle (say $C_u^1$) has three distinguished vertices $u_R',
    u_G', u_B'$ at distance $\ell$ from each other. The second (called
    $C_u^2$) has similar distinguished vertices $u_R'', u_G'', u_B''$ at
    distance $\ell$ from each other.
  \item There are three more ``interface'' vertices $u_R, u_G,
    u_B$. Vertex $u_R$ is connected to $u_R'$ and $u_R''$, similarly for
    $u_G$ and $u_B$.
  \item There is a special ``switch'' vertex $s_u$, which is connected
    to all three of $\{u_R, u_G, u_B\}$. Call these edges the
    \emph{switch} edges.
  \end{OneLiners}
  Due to the two odd cycles, every perfect matching in $X_u$ has the
  structure that one of the interface vertices is matched to some vertex
  in $C_u^1$, another to a vertex in $C_u^2$ and the third to the switch
  $s_u$.  We think of the subscript of the vertex matched to $s_u$ as
  the color assigned to the vertex $u$.

  At every odd time step $t \in T$, the only allowed edges are those
  within the gadgets $\{X_u\}_{u \in V}$: i.e., all the holding costs
  for edges within the gadgets is zero, and all edges between gadgets
  have holding costs $\infty$. This is called the ``steady state''.

  At every even time step $t$, for some matching $M_t \sse E$ of the
  graph, we move into a ``test state'', which intuitively tests whether
  the edges of a matching $M_t$ have been properly colored. We do this
  as follows. For every edge $(u,v) \in M_t$, the switch edges in $X_u,
  X_v$ become unavailable (have infinite holding costs). Moreover, now
  we allow some edges that go between $X_u$ and $X_v$, namely the edge
  $(s_u, s_v)$, and the edges $(u_i, v_j)$ for $i,j \in \{R,G,B\}$ and
  $i \neq j$. Note that any perfect matching on the vertices of $X_u
  \cup X_v$ which only uses the available edges would have to match
  $(s_u, s_v)$, and one interface vertex of $X_u$ must be matched to one
  interface vertex of $X_v$. Moreover, by the structure of the allowed
  edges, the colors of these vertices must differ. (The other two
  interface vertices in each gadget must still be matched to their odd
  cycles to get a perfect matching.) Since the graph has bounded degree,
  we can partition the edges of $G$ into a constant number of matchings
  $M_1, M_2, \ldots, M_{\Delta}$ for some $\Delta = O(1)$ (using Vizing's
  theorem). Hence, at time step $2\tau$, we test the edges of the
  matching $M_\tau$. The number of timesteps
  is $T = 2\Delta$, which is a constant.
  \begin{figure}[h]
    \centering
    \includegraphics[scale=0.5]{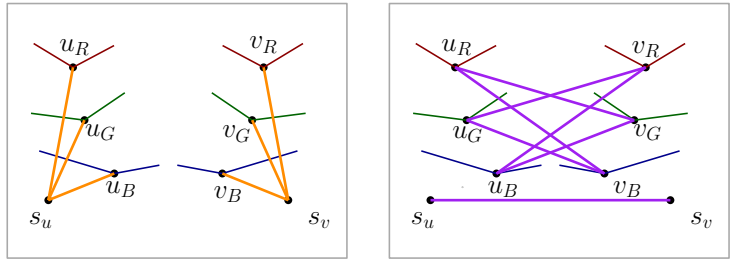}
    \caption{On the left, the steady-state edges incident to the
      interface and switch vertices of edge $(u,v)$. The
      test-state edges are on the right.}
    \label{fig:gadget2}
  \end{figure}

  Suppose the graph $G$ was indeed $3$-colorable, say $\chi: V \to
  \{R,G,B\}$ is the proper coloring. In the steady states, we choose a
  perfect matching within each gadget $X_u$ so that $(s_u, u_{\chi(u)})$
  is matched. In the test state $2t$, if some edge $(u, v)$ is in the
  matching $M_t$, we match $(s_u, s_v)$ and $(u_{\chi(u)},
  v_{\chi(v)})$. Since the coloring $\chi$ was a proper coloring, these
  edges are present and this is a valid perfect matching using only the
  edges allowed in this test state. Note that the only changes are that
  for every test edge $(u,v)\in M_t$, the matching edges $(s_u,
  u_{\chi(u)})$ and $(s_v, v_{\chi(v)})$ are replaced by $(s_u, s_v)$
  and $(u_{\chi(u)}, v_{\chi(v)})$. Hence the total acquisition cost
  incurred at time $2t$ is $2|M_t|$, and the same acquisition cost is
  incurred at time $2t+1$ to revert to the steady state. Hence the total
  acquisition cost, summed over all the timesteps, is $4|E|$.

  Suppose $G$ is not $3$-colorable. We claim that there exists vertex $u
  \in U$ such that the interface vertex not matched to the odd cycles is
  different in two different timesteps---i.e., there are times $t_1,
  t_2$ such that $u_i$ and $u_j$ (for $i \neq j$) are the states. Then
  the length of the augmenting path to get from the perfect matching at
  time $t_1$ to the perfect matching at $t_2$ is at least $\ell$. Now if
  we set $\ell = n^{2/\eps}$, then we get a total acquisition cost of at
  least $n^{2/\eps}$ in this case.

  The size of the graph is $N := O(n\ell) = O(n^{1+2/\eps})$, so the gap
  is between $4|E| = O(n) = O(N^{\eps})$ and $\ell = N^{1-\eps}$. This
  proves the claim.
\end{proof}
\fi

\myspace{-0.3in}
\section{Conclusions}
\myspace{-0.1in}

In this paper we studied multistage optimization problems: an
optimization problem (think about finding a minimum-cost spanning tree
in a graph) needs to be solved repeatedly, each day a different set of
element costs are presented, and there is a penalty for changing the
elements picked as part of the solution. Hence one has to hedge between
sticking to a suboptimal solution and changing solutions too rapidly.
We present online and offline algorithms when the optimization problem
is maintaining a base in a matroid. We show that our results are optimal
under standard complexity-theoretic assumptions. We also show that the problem of
maintaining a perfect matching becomes impossibly hard.

Our work suggests several directions for future research. It is natural
to study other combinatorial optimization problems, both polynomial time
solvable ones such shortest path and min-cut, as well NP-hard ones such
as min-max load balancing and bin-packing in this multistage framework
with acquisition costs. Moreover, the approximability of the {\em bipartite} matching maintenance, as well as matroid intersection maintenance remains open.  Our hardness results for the matroid problem
hold when edges have $\{0,1\}$ acquisition costs. The unweighted version
where all acquisition costs are equal may be easier; we currently know
no hardness results, or sub-logarithmic approximations for this useful
special case.


\bibliographystyle{abbrv}
{\bibliography{basepackingX}}

\appendix

\section{Lower Bounds: Hardness and Gap Results}

\subsection{Hardness for Time-Varying Matroids}
\label{sec:time-varying}

An extension of \MMM/\MSM problems is to the case when the set of
elements remain the same, but the matroids change over time. Again the
goal in \MMM is to maintain a matroid base at each time.

\begin{theorem}
  \label{thm:diff-matrs-wpb}
  The \MMM problem with different matroids is NP-hard to approximate
  better than a factor of $\Omega(T)$, even for partition matroids, as
  long as $T \ge 3$.
\end{theorem}

\begin{proof}
  The reduction is from 3D-Matching (3DM). An instance of 3DM has three
  sets $X, Y, Z$ of equal size $|X| = |Y| = |Z| = k$, and a set of
  hyperedges $E \sse X \times Y \times Z$. The goal is to choose a set
  of disjoint edges $M \sse E$ such that $|M| = k$.

  First, consider the instance of \MMM with three timesteps $T = 3$. The
  universe elements correspond to the edges. For $t=1$, create a
  partition with $k$ parts, with edges sharing a vertex in $X$ falling
  in the same part. The matroid $\M_1$ is now to choose a set of
  elements with at most one element in each part. For $t=2$, the
  partition now corresponds to edges that share a vertex in $Y$, and for
  $t=3$, edges that share a vertex in $Z$. Set the movement weights
  $w(e) = 1$ for all edges.

  If there exists a feasible solution to 3DM with $k$ edges, choosing
  the corresponding elements form a solution with total weight $k$. If
  the largest matching is of size $(1 - \varepsilon)k$, then we must pay
  $\Omega(\varepsilon\, k)$ extra over these three timesteps. This gives
  a $k$-vs-$(1+\Omega(\varepsilon))k$ gap for three timesteps.

  To get a result for $T$ timesteps, we give the same matroids
  repeatedly, giving matroids $\M_{t \pmod 3}$ at all times $t \in
  [T]$. In the ``yes'' case we would buy the edges corresponding to the
  3D matching and pay nothing more than the initial $k$, whereas in the
  ``no'' case we would pay $\Omega(\varepsilon k)$ every three
  timesteps. Finally, the APX-hardness for 3DM~\cite{Kann:1991:MBM:105391.105396} gives the
  claim.
\end{proof}

The time-varying \MSM problem does admit an $O(\log rT)$ approximation,
as the randomized rounding (or the greedy algorithm) shows. However, the
equivalence of \MMM and \MSM does not go through when the matroids
change over time.

The restriction that the matroids vary over time is essential for the
NP-hardness, since if the partition matroid is the same for all times,
the complexity of the problem drops radically.
\begin{theorem}
  \label{thm:partition}
  The \MMM problem with partition matroids can be solved in polynomial
  time.
\end{theorem}

\begin{proof}
  The problem can be solved using min-cost flow. Indeed, consider the
  following reduction.  Create a node $v_{et}$ for each element $e$ and
  timestep $t$. Let the partition be $E = E_1 \cup E_2 \cup \ldots \cup
  E_r$. Then for each $i \in [r]$ and each $e, e' \in E_i$, add an arc
  $(v_{et}, v_{e',t+1})$, with cost $w(e') \cdot \mathbf{1}_{e \neq
    e'}$. Add a cost of $c_t(e)$ per unit flow through vertex
  $v_{et}$. (We could simulate this using edge-costs if needed.)
  Finally, add vertices $s_1, s_2, \ldots, s_r$ and source $s$. For each
  $i$, add arcs from $s_i$ to all vertices $\{v_{e1}\}_{e \in E_i}$ with
  costs $w(e)$. All these arcs have infinite capacity. Now add unit
  capacity edges from $s$ to each $s_i$, and infinite capacity edges
  from all nodes $v_{eT}$ to $t$.

  Since the flow polytope is integral for integral capacities, a flow of
  $r$ units will trace out $r$ paths from $s$ to $t$, with the elements
  chosen at each time $t$ being independent in the partition matroid,
  and the cost being exactly the per-time costs and movement costs of
  the elements. Observe that we could even have time-varying movement
  costs.  Whereas, for graphical matroids the problem is $\Omega(\log
  n)$ hard even when the movement costs for each element do not change
  over time, and even just lie in the set $\{0,1\}$.
\end{proof}

Moreover, the restriction in Theorem~\ref{thm:diff-matrs-wpb} that $T
\ge 3$ is also necessary, as the following result shows.
\begin{theorem}
  \label{thm:two}
  For the case of two rounds (i.e., $T = 2$) the \MSM problem can be
  solved in polynomial time, even when the two matroids in the two
  rounds are different.
\end{theorem}

\begin{proof}
  The solution is simple, via matroid intersection. Suppose the matroids
  in the two timesteps are $\M_1 = (E, \I_1)$ and $\M_2 =
  (E,\I_2)$. Create elements $(e, e')$ which corresponds to picking
  element $e$ and $e'$ in the two time steps, with cost $c_1(e) +
  c_2(e') + w_e + w_{e'}\mathbf{1}_{e \neq e'}$. Lift the matroids
  $\M_1$ and $\M_2$ to these tuples in the natural way, and look for a
  common basis.
\end{proof}




\subsection{Lower Bound for Deterministic Online Algorithms}
\label{sec:lbd-deterministic-online}

We note that deterministic online algorithms cannot get any non-trivial guarantee for the \MMM problem, even in the simple case of a $1$-uniform matroid. This is related to the lower bound for deterministic algorithms for paging. Formally, we have the 1-uniform matroid on $m$ elements, and $T=m$. All acquisition costs $a(e)$ are 1. In the first period, all holding costs are zero and the online algorithm picks an element, say $e_1$. Since we are in the non-oblivious model,the algorithm knows $e_1$ and can in the second time step, set $c_2(e_1)=\infty$, while leaving the other ones at zero. Now the algorithm is forced to move to another edge, say $e_2$, allowing the adversary to set $c_3(e_2)=\infty$ and so on. At the end of $T=m$ rounds, the online algorithm is forced to incur a cost of 1 in each round, giving a total cost of $T$. However, there is still an edge whose holding cost was zero throughout, so that the offline OPT is 1. Thus against a non-oblivious adversary, any online algorithm must incur a $\Omega(\min(m,T))$ overhead.

\subsection{An $\Omega(\min(\log T, \log \frac{a_{max}}{a_{min}}))$ LP
  Integrality Gap}
\label{sec:int-gap-matroids}

In this section, we show that if the aspect ratio of the movement costs is not bounded, the linear program has a $\log T$ gap, even when $T$ is exponentially larger than $m$. We present an instance where $\log T$ and $\log \frac{a_{max}}{a_{min}}$ are about $r$ with $m=r^2$, and the linear program has a gap of $\Omega(\min(\log T, \log \frac{a_{max}}{a_{min}}))$. This shows that the $O(\min(\log T,\log\frac{a_{max}}{a_{min}}))$ term in our rounding algorithm is unavoidable.

The instance is a graphical matroid, on a graph $G$ on $\{v_0,v_1,\ldots,v_n\}$, and $T={n \choose \frac{n}{2}} = 2^{O(n)}$. The edges $(v_0,v_i)$ for $i\in[n]$ have acquisition cost $a(v_0,v_i)=1$ and holding cost $c_t(v_0,v_i)=0$ for all $t$. The edges $(v_i,v_j)$ for $i,j \in [n]$ have acquisition cost $\frac{1}{nT}$ and have holding cost determined as follows: we find a bijection between the set $[T]$  and the set of partitions $(U_t,V_t) $ of $\{v_1,\ldots,v_n\}$ with each of $U_t$ and $V_t$ having size $\frac{n}{2}$ (by choice of $T$ such a bijection exists, and can be found e.g. by arranging the $U_t$'s in lexicographical order.) . In time step $t$, we set $c_t(e) = 0$ for $e \in (U_t\times U_t) \cup (V_t \times V_t)$, and $c_t(e)=\infty$ for all $e \in U_t \times V_t$.

First observe that no feasible integral solution to this instance can pay acquisition cost less than $\frac{n}{2}$ on the $(v_0,v_i)$ edges. Suppose that the solution picks edges $\{(v_0,v_i): v_i \in U_{sol}\}$ for some set $U_{sol}$ of size at most $\frac{n}{2}$. Then any time step $t$ such that $U_{sol} \subseteq U_t$, the solution has picked no edges connecting $v_0$ to $V_t$, and all edges connecting $U_t$ to $V_t$ have infinite holding cost in this time step. This contradicts the feasibility of the solution. Thus any integral solution has cost $\Omega(n)$.

Finally, we show that on this instance, (\ref{eq:lp2}) from Section~\ref{sec:lp-round}, has a feasible solution of cost $O(1)$. We set $y_t(v_0,v_i) = \frac{2}{n}$ for all $i \in [n]$, and set $y_t(v_i,v_j) = \frac{2}{n}$ for $(v_i,v_j) \in (U_t \times U_t) \cup (V_t \times V_t)$. It is easy to check that $z_t=y_t$ is in the spanning tree polytope for all time steps $t$. Finally, the total acquisition cost is at most $n \cdot 1 \cdot \frac{2}{n}$ for the edges incident on $v_0$ and at most $T \cdot n^2 \cdot \frac{1}{nT} \cdot \frac{2}{n}$ for the other edges, both of which are $O(1)$. The holding costs paid by this solution is zero. Thus the LP has a solution of cost $O(1)$

The claim follows.


\section{The Greedy Algorithm}
\label{sec:greedy}

\newcommand{\ben}{\textsf{ben}}
\newcommand{\MST}{\textsf{mst}}

The greedy algorithm for \MSM is the natural one.
We consider the interval view of the problem (as in
Section~\ref{sec:intervals}) where each element only has acquisition
costs $a(e)$, and can be used only in some interval $I_e$. Given a
current subset $X \sse E$, define $X_t := \{ e' \in X \mid I_{e'} \ni t
\}$.  The benefit of adding an element $e$ to $X$ is
\[ \ben_X(e) = \sum_{t \in I_e} (\rank(X_t \cup \{e\}) - \rank(X_t)) \]
and the greedy algorithm repeatedly picks an element $e$ maximizing
$\ben_X(e)/a(e)$ and adds $e$ to $X$. This is done until $\rank(X_t) =
r$ for all $t \in [T]$.

Phrased this way, an $O(\log T)$ bound on the approximation ration follows from Wolsey~\cite{Wol82}. We next give an alternate dual fitting proof. We do not know of an instance with uniform acquisition costs where greedy does not give a constant factor approximation. The dual fitting approach may be useful in proving a better approximation bound for this special case.

The natural LP is:
\begin{align}
  P := \min \ts \sum_e a(e) \cdot x_e & \tag{LP1} \label{eq:lp1} \\
  \text{s.t.~~~} \{ z_{et} \}_e &\in \P_B(\M) \qquad\qquad\qquad \forall
  t \notag \\
    z_{et} &\leq x(e) \qquad\qquad\qquad\forall e, \forall t \in I_e \notag\\
    x_e &\geq 0 \qquad\qquad\qquad\forall e \notag \\
    z_{et} &\geq 0 \qquad\qquad\qquad\forall e, \forall t \in I_e \notag
\end{align}
where the polytope $\P_B(\M)$ is the base polytope of the matroid $\M$.

Using Lagrangian variables $\beta_{et} \geq 0$ for each $e$ and $t \in
I_e$, we write a lower bound for $P$ by
\begin{align*}
  D(\beta) := \min \ts \sum_e a(e) \cdot x_e &+ \sum_{e,t \in I_e} \beta_{et}(z_{et} - x_e) \\
  \text{s.t.~~~} z_{et} &\in \P_B(\M) \qquad\qquad\qquad \forall t \\
    x_e, z_{et} &\geq 0
\end{align*}
which using the integrality of the matroid polytope can be rewritten as:
\begin{align*}
  \min_{x \geq 0} \ts \sum_e x_e \big(a(e) - \sum_{e,t \in I_e}
  \beta_{et}\big) &+ \ts \sum_{t}
  \MST(\beta_{et}).
\end{align*}
Here, $\MST(\beta_{et})$ denotes the cost of the minimum weight base at
time $t$ according to the element weights $\{\beta_{et}\}_{e \in E}$,
where the available elements at time $t$ is $E_t = \{ x \in E \mid t \in
I_e \}$. The best lower bound is:
\begin{align*}
  D := \max \ts \sum_{t}  \MST(\beta_{et}) & \\
  \text{s.t.~~~} \ts \sum_{t \in I_e} \beta_{et} \leq a(e) \\
  \beta_{et} \geq 0.
\end{align*}

The analysis of greedy follows the dual-fitting proofs of~\cite{Chvatal, FNW-II}.
\begin{theorem}
  \label{thm:greedy}
  The greedy algorithm outputs an $O(\log |I_{\max}|)$-approximation to
  \MSM, where $|I_{\max}|$ is the length of the longest interval that an
  element is alive for. Hence, it gives an $O(\log T)$-approximation.
\end{theorem}

\begin{proof}
  For the proof, consider some point in the run of the greedy algorithm
  where set $X$ of elements has been picked. We show a setting of duals
  $\beta_{et}$ such that
  \begin{OneLiners}
  \item[(a)] the dual value equals the current primal
    cost $\sum_{e \in X} a(e)$, and
  \item[(b)] the constraints are nearly satisfied, namely $\sum_{t \in
      I_e} \beta_{et} \leq a(e) \log {|I_e|}$ for every $e \in E$.
  \end{OneLiners}
  It is useful to maintain, for each time $t$, a \emph{minimum weight
    base} $B_t$ of the subset $\spn(X_t)$ according to weights
  $\{\beta_{et}\}$.  Hence the current dual value equals $\sum_t \sum_{e
    \in B_t} \beta_{et}$. We start with $\beta_{et} = 0$ and $X_t = B_t
  = \emptyset$ for all $t$, which satisfies the above properties.

  Suppose we now pick $e$ maximizing $\ben_X(e)/a(e)$ and get new set
  $X' := X \cup \{e\}$. We use $X'_t := \{ e' \in X' \mid I_{e'} \ni t
  \}$ akin to our definition of $X_t$. Call a timestep $t$
  ``interesting'' if $\rank(X'_t) = \rank(X_t) + 1$; there are
  $\ben_X(e)$ interesting timesteps. How do we update the duals? For $e'
  \in \spn(X'_t) \setminus \spn(X_t)$, we set $\beta_{e't} \gets
  a(e)/\ben_X(e)$. Note the element $e$ itself satisfies the condition
  of being in $\spn(X'_t) \setminus \spn(X_t)$ for precisely the
  interesting timesteps, and hence $\sum_{t \text{ interesting}}
  \beta_{et} = (a(e)/\ben_X(e)) \cdot \ben_X(e) = a(e)$. For each
  interesting $t \in I_e$, define the base $B'_t \gets B_t + e$; for all
  other times set $B'_t \gets B_t$. It is easy to verify that $B'_t$ is
  a base in $\spn(X'_t)$. But is it a min-weight base?  Inductively
  assume that $B_t$ was a min-weight base of $\spn(X_t)$; if $t$ is not
  interesting there is nothing to prove, so consider an interesting $t$.
  All the elements in $\spn(X'_t) \setminus \spn(X_t)$ have just been
  assigned weight $\beta_{e't} = a(e)/\ben_X(e)$, which by the
  monotonicity properties of the greedy algorithm is at least as large
  as the weight of any element in $\spn(X_t)$. Since $e$ lies in
  $\spn(X'_t) \setminus \spn(X_t)$ and is assigned value $\beta_{et} =
  a(e)/\ben_X(e)$, it cannot be swapped with any other element in
  $\spn(X'_t)$ to improve the weight of the base, and hence $B'_t = B_t
  + e$ is an min-weight base of $\spn(X'_t)$.

  It remains to show that the dual constraints are approximately
  satisfied. Consider any element $f$, and let $\lambda = |I_f|$. The
  first step where we update $\beta_{ft}$ for some $t \in I_f$ is when
  $f$ is in the span of $X_t$ for some time $t$. We claim that
  $\beta_{ft} \leq a(f)/\lambda$. Indeed, at this time $f$ is a
  potential element to be added to the solution and it would cause a
  rank increase for $\lambda$ time steps. The greedy rule ensures that
  we must have picked an element $e$ with weight-to-coverage ratio at
  most as high. Similarly, the next $t$ for which $\beta_{ft}$ is
  updated will have $a(f)/(\lambda - 1)$, etc. Hence we get the sum
  \[ \sum_t \beta_{ft} \leq a(f) \left( \frac1{|I_f|} + \frac1{|I_f| -
      1} + \cdots + 1 \right) \leq a(f) \times O(\log |I_f|). \]
  Since each element can only be alive for all $T$ timesteps, we get the
  claimed $O(\log T)$-approximation.
\end{proof}

Note that the greedy algorithm would solve $\MSM$ even if we had a
different matroid $\M_t$ at each time $t$. However, the equivalence of
\MMM and \MSM no longer holds in this setting, which is not surprising
given the hardness of Theorem~\ref{thm:diff-matrs-wpb}.

\end{document}